\documentclass{article}
\usepackage[utf8]{inputenc}
\usepackage{dsfont}
\usepackage{amsmath,amssymb,amsthm,mathtools}
\usepackage{xcolor}
\usepackage[onehalfspacing]{setspace}
\usepackage[nohead]{geometry}
\usepackage{latexsym}
\usepackage{array}
\usepackage{enumitem}
\usepackage[round,comma]{natbib}
\usepackage{pdfsync}

\definecolor{uniblau}{cmyk}{1,.9,0,0}

\usepackage{hyperref} 
\hypersetup{ 
     backref=true,    
     pagebackref=true,
     citecolor=uniblau,
     hyperindex=true, 
     colorlinks=true, 
     breaklinks=true, 
     urlcolor= uniblau,
     linkcolor= uniblau,
     bookmarks=true,  
     bookmarksopen=true
     pdftitle={Signaling via Money},
     pdfauthor={Olivier Bos, Martin Pollrich},  
     pdfsubject={}   
}
\usepackage{cleveref}
\usepackage[multiple]{footmisc}

\newtheorem{thm}{Theorem}

\Crefname{prop}{Proposition}
{Propositions}

\theoremstyle{definition}
\newtheorem{example}{Example}

\newtheorem{rem}{Remark}

\def\uv{\underline v}
\def\ov{\overline v}

\definecolor{beamer@blendedblue}{rgb}{0.2,0.2,0.7}

\definecolor{beamer@blendedred}{RGB}{214,39,40}

\title{Signaling via Money\footnote{We gratefully thank Michael Greinecker, Patrick Lahr, and Ernesto Rivera Mora for helpful comments and discussions. Financial support from the ANR SIGNAL (ANR-19-CE26-0009) is gratefully acknowledged.}}\author{{\sc Olivier Bos}\footnote{Universit\'e Paris-Saclay, ENS Paris-Saclay, Centre for Economics at Paris-Saclay. E-mail: olivier.bos@ens-paris-saclay.fr.} \; and {\sc Martin Pollrich}\footnote{Universit\'e Paris-Saclay, ENS Paris-Saclay, Centre for Economics at Paris-Saclay. E-mail: martin.pollrich@ens-paris-saclay.fr.}}

\date{October 2026}

\begin{document}

\bibliographystyle{agsm}\addtocontents{toc}{\protect\thispagestyle{plain}}

\maketitle

\begin{abstract}
We study mechanisms in which agents' preferences depend both on an object allocation and on a third party's beliefs about private information. A mechanism determines an allocation, transfers, and sends signals to the aftermarket. We compare two informational environments: (i) private transfers, where the aftermarket observes allocations and signals, and (ii) public transfers, where transfers are publicly observable. Our main result shows that any equilibrium outcome implementable with hidden transfers can be approximated arbitrarily well with public transfers. Thus, transfer observability generates at most a vanishing informational distortion and helps rationalize the common assumption that transfers remain hidden from aftermarkets.

\end{abstract}

\section{Introduction}

In many mechanism design environments the agents interact with an aftermarket after execution of the initial mechanism. Examples are auctions with resale, where the winning bidder has the opportunity to sell the object further to a third party or takeover auctions, where bidding firms remain competitors in their respective markets. In many of these aftermarkets the agents' private information remains decision-relevant. For instance a winning bidder who offers the good for resale becomes a privately informed seller, and a firm's bidding in a takeover auction can be used by its competitors as a proxy for competitiveness. Controlling the information flow to the aftermarket therefore becomes a central feature of the mechanism design problem, because each agent's preferences depend not only on the outcome of the mechanism, but also on the outcome of the aftermarket. Put differently, the mechanism design problem is augmented by an information design problem. Importantly, there are two channels of information flow the designer must consider. First, she can directly communicate with the aftermarket by sending signals. Of course, these signals may be correlated with the information revealed by the agents to the mechanism. For instance the designer can fully reveal the agents' messages. Second, there is implicit information revelation from the observable components of the allocation implemented by the mechanism. For example, the fact that an agent was allocated an object is informative about that agent's private information given the allocation rule of the mechanism. In this paper, we focus on the constraints imposed on the mechanism designer by this latter channel of implicit information revelation.

The focus of this paper is the impact of implicit information revelation from observing the agents' monetary payments. Formally, we compare two informational environments. Under \emph{private transfers}, the aftermarket observes only the physical allocation (e.g. who gets the object) but not the payments made to or from the mechanism. In that environment, the mechanism may also communicate with the aftermarket through explicit signals. Under \emph{public transfers}, instead, payments themselves become part of the publicly observed outcome. We ask whether and how this observability of transfers constrains the mechanism designer.

To see why there is no trivial answer, reconsider the standard analytical procedure in mechanism design. Typically, the designer invokes a revelation principle and restricts attention to a direct mechanism. Such a mechanism associates an allocation with every type profile, allowing the aftermarket to infer precise information about agents' types by reverse engineering from the allocation rule. In particular, in many problems the transfers allow perfect inference about agents' types. While this is without harm if transfers remain hidden from the aftermarket, it may fundamentally alter the informational environment faced by the
aftermarket.\footnote{Yet, \cite{ErnestoMora} allows transfers to be observable, non-observable, or partially observable in extended direct mechanisms, taking the observability regime as given rather than studying its implications for implementability.}

With public transfers the designer faces the challenge of controlling the information flow to the aftermarket from the transfers, while at the same time these transfers remain crucial for incentive compatibility. Our main result shows that the apparent tension can be resolved approximately. Any equilibrium outcome implementable with private transfers can be replicated arbitrarily closely with public transfers. This provides a novel perspective on the role of observability in mechanisms with aftermarkets and belief-dependent preferences. While public transfers may reveal additional information, the resulting distortions can be made arbitrarily small as long as continuation payoffs depend on observables only through the posterior beliefs they induce. Importantly, this formulation already allows continuation payoffs to arise from strategic responses by outsiders, provided such responses depend only on posterior beliefs. The analysis therefore sheds new light on recent work on signaling in mechanisms followed by either an aftermarket \citep{BosPollrich,Dworczak19} or an after-game \citep{ErnestoMora}.

\medskip

The paper is organized as follows. \Cref{Section:Litterature} examines the related literature. \Cref{Section:Model} introduces the formal setting. \Cref{Section:Analysis} establishes our implementability equivalence result and illustrates how it extends previous contributions. Even when transfers are observable, equilibrium outcomes achievable with private transfers can be approximated arbitrarily closely with public transfers. \Cref{Section:Conclusion} concludes.

\section{Related Litterature}\label{Section:Litterature}

\cite{Dworczak19} studies mechanisms with aftermarkets and shows that implementability can be characterized by restricting attention to cutoff mechanisms, in which the information revealed about the winner depends only on the losers' bids. \cite{BosPollrich} investigate optimal auctions with reputational concerns and demonstrate how disclosure rules affect revenue through the bidders' endogenous outside options. They derive a revenue decomposition and show that optimality and revenue equivalence rely on full disclosure or on the observability of participation decisions, and extend these insights to settings with non-linear reputational concerns. \cite{ErnestoMora} develops a revelation principle for settings with belief-dependent preferences, introducing extended direct mechanisms that guarantee implementability when utilities depend on higher-order beliefs. We establish that equilibrium outcomes implementable with hidden transfers can be approximated arbitrarily well using public transfers. In this respect, our result broadens the scope of implementability analysis beyond the cutoff characterization of \cite{Dworczak19} and complements \cite{BosPollrich} by showing that transfer observability generates at most a vanishing implementability distortion. It also provides an approximate counterpart to the revelation principle of \cite{ErnestoMora} within the subclass of first-order belief-based environments.

\medskip

More broadly, the paper relates to the literature on mechanism design with information disclosure policies. \cite{CalzolariPavan06a, CalzolariPavan06b} study optimal allocation and information released in contracting environments where the agent interacts in an aftermarket, and identify conditions under which no disclosure to the aftermarket becomes optimal. \cite{EilatEliazMu24} study privacy-preserving auctions, where they define privacy loss as the mutual information between payments and types. They consider mechanisms in which the winner's identity and her payment are disclosed, and characterize those minimizing the winner's privacy loss given payment observability. While their analysis treats observability as a dimension shaping the mechanism design problem, our contribution shows that transfer observability imposes at most a vanishing implementability distortion. \cite{HauptHitzig} analyze implementability under constraints on how much private information the designer learns, minimizing redundant information revelation.  Our focus instead is on the robustness of implementability when transfers become observable to a third party, even though such observability affects the information available to an external audience rather than to the designer. Other papers provide positive comparative analyses of auction formats and disclosure policies \citep{BosTruyts17, GiovannoniMarkis14}, including settings where the auction is followed by oligopolistic competition (see for example \cite{DasVarma2003} and \cite{Goeree2003}).

\medskip

Our result is naturally related to the literature on the implementability of social choice functions in direct incentive-compatible mechanisms, initiated by \cite{Gibbard} and \cite{Myerson81}. More loosely related work has examined credibility and transparency in dynamic mechanisms. For example \cite{AkbarpourLi} study the credibility of dynamic auctions, while \cite{HakimovRaghavan} and \cite{Moeller} analyze how transparency requirements affect implementability. While our analysis focuses on mechanisms with a third party and the observability of transfers, these contributions highlight how credibility and transparency shape implementable outcomes in dynamic settings.

\section{Preliminaries}\label{Section:Model}
A designer allocates a single object, for which she has no value, among $n\geq 1$ agents with valuations $(v_1,\ldots,v_n)$ for the object. 

Agent $i$'s valuation for the object is denoted $v_i$, and independently drawn according to a cdf with support on $V_i:\equiv [\uv_i, \ov_i] \subset \mathbb{R}_{+}$.\footnote{There is no need to provide more details about this distribution. Moreover, for expositional simplicity, Example 2 below uses binary valuations.} Let $K=\{0,1,\ldots,n\}$ denote the set of feasible allocations, where $0$ stands for the seller keeping the object. A social choice function $(k,\mathbf t)$ consists of an allocation rule $k:V_1\times\cdots\times V_n\to K$ and a vector of transfers $\mathbf t=(t_1,\ldots,t_n)$, where $t_i:V_1\times\cdots\times V_n\to\mathbb R$. Implementation is understood in the standard interim sense: the induced equilibrium outcomes coincide with the target social choice function in terms of allocations and expected transfers.

After the mechanism, a Bayesian game unfolds between the agents and a third party (or general public), whose beliefs may depend on the observed outcome. While in some applications only the winner's beliefs are payoff-relevant, our formulation allows beliefs about all agents to matter. This audience forms posterior beliefs represented by a joint distribution $\mathbf F\in\Delta(V_1\times\cdots\times V_n),$ where $\Delta(V_1\times\cdots\times V_n)$ denotes the set of probability distributions over the agents' types. In our setting, agent $i$'s utility is affected by the aftermarket's beliefs about her own valuation through the mapping $\Phi_i:V_i\times K\times \Delta(V_1\times\cdots\times V_n)\to\mathbb R$. For simplicity, $\Phi_i(\cdot,k,\mathbf F)$ is assumed to be monotone and continuous. The formulation also allows continuation payoffs to arise from strategic responses by outsiders. Provided such responses depend only on posterior beliefs induced by observables, they can be represented in reduced form through $\Phi_i(v_i,k,\mathbf F)$. This is a reduced-form representation: strategic responses are incorporated through the posterior beliefs induced by observables.

Agent $i$'s utility from the allocation $(k,\mathbf t)$ is
$$
v_i\mathds 1_{\{k=i\}}  - t_i + \Phi_i(v_i,k,\mathbf F).
$$
\paragraph{Mechanisms with Signals.} A \emph{signaling mechanism} $\Gamma^s$ consists of message sets $(M_1,\ldots,M_n)$ and an allocation rule 
$$\sigma : M_1 \times \cdots \times M_n \times [0,1] \rightarrow K \times \mathbb R^n \times S,$$ 
where $S$ denotes the set of signals. Agents send messages to the mechanism. Conditional on the vector of messages and the realization of a random draw from $\mathcal U[0,1]$, the mechanism selects an allocation $(k,\mathbf t)$ and sends a signal $s\in S$ to the aftermarket. The public observes the allocation $k$ and the realized signal $s$, but neither transfers nor messages, and forms posterior beliefs based on these observations.

\paragraph{Mechanisms with Public Transfers.} A \emph{public transfers mechanism} $\Gamma^p$ consists of message sets $(M_1,\ldots,M_n)$ and an allocation rule 
$$\sigma : M_1 \times \cdots \times M_n \times [0,1] \rightarrow K \times \mathbb R^n.$$
Agents send messages to the mechanism. Conditional on the vector of messages and the realization of a random draw from $\mathcal U[0,1]$, the mechanism selects an allocation $(k,\mathbf t)$. In contrast to signaling mechanisms, the entire allocation $(k,\mathbf t)$ is publicly observable, while messages remain private.

\medskip

A social choice function $(k,\mathbf t)$ is said to be implementable if there exists a mechanism together with an equilibrium such that the induced equilibrium outcome coincides with $(k,\mathbf t)$ in terms of allocations and expected transfers.

\section{Results}\label{Section:Analysis}

We first state our main result and then provide two illustrations showing how it relates to previous work. A final example illustrates how publicly observed transfers can replicate arbitrarily closely the informational role of signals.

\begin{thm}[Approximate Implementability Equivalence]\label{prop1}
The following statements hold:
\begin{enumerate}[label=(\roman*)]
\item A social choice function $(k,\mathbf t)$ that is implementable with a signaling mechanism is approximately implementable with a public transfers mechanism.

That is, for every $\varepsilon>0$, there exists a public transfers mechanism admitting an $\varepsilon$-Bayes-Nash equilibrium such that the induced allocation probabilities, expected transfers, and interim expected payoffs are within $\varepsilon$ of those induced by $(k,\mathbf t)$.
\item Any social choice function $(k,\mathbf t)$ implementable with a public transfers mechanism is implementable with a signaling mechanism.
\end{enumerate}
\end{thm}

\begin{proof}

We first prove (ii). Fix a social choice function $(k,\mathbf t)$ that is implementable with a public transfers mechanism $\Gamma^p$. Construct a signaling mechanism $\Gamma^s$ with the same message spaces, allocation rule, and transfer rule as $\Gamma^p$, and let $\Gamma^s$ send a signal equal to the realized transfer vector.

The aftermarket therefore observes exactly the same information in $\Gamma^s$ as in $\Gamma^p$, so posterior beliefs are unchanged. Moreover, agents face the same allocation, the same transfers, and the same continuation payoffs after every message profile. Hence any equilibrium of $\Gamma^p$ remains an equilibrium of $\Gamma^s$, and the same social choice function is implemented.

\medskip

We now prove (i). Fix $\varepsilon>0$ and a social choice function $(k,\mathbf t)$ implemented by a signaling mechanism $\Gamma^s$. Let  $Q_{i}(v^\prime)$ and $T_{i}(v^\prime)$  denote the interim allocation probability and interim expected transfer induced when agent $i$ mimics the equilibrium behavior prescribed for type $v^\prime$. Let $k(v^\prime,V_{-i})$ and $s(v^\prime,V_{-i})$ denote the allocation and signal induced by the equilibrium of $\Gamma^s$ when agent $i$ mimics the equilibrium behavior prescribed for type $v^\prime$ while the other agents follow their equilibrium behavior. If the true type of agent $i$ is $v$ and she mimics the equilibrium behavior prescribed for type $v^\prime$, her interim payoff in $\Gamma^s$ is

\begin{equation*}
  U_i(v,v^\prime) =  v Q_i(v^\prime)  -  T_i(v^\prime)  +  \mathbb E_{V_{-i}}
  \left[\Phi_i\left(v, k(v^\prime,V_{-i}),  \mathbf F(k(v^\prime,V_{-i}),s(v^\prime,V_{-i}))  \right)  \right].
\end{equation*}

Since the equilibrium of $\Gamma^s$ implements $(k,\mathbf t)$, no type has a profitable deviation toward the equilibrium behavior prescribed for another type. Hence,
\begin{equation*}
  U_i(v,v) \geq U_i(v,v^\prime)  \text{ for all } i,\ v,\ v^\prime.
\end{equation*}

Fix $\delta\in(0,1)$. Its value will be fixed later as a function of $\varepsilon$. We construct a public transfers mechanism that approximates $\Gamma^s$. The constructed public transfers mechanism is a direct mechanism in which agents report types. Given a report profile, it uses the allocation and signal induced by the equilibrium of $\Gamma^s$ for the corresponding type profile, and then assigns transfers according to the rule constructed below. With probability $1-\delta$, it assigns the transfer vector

$$\widehat{\mathbf t}(k,s)=f(k,s),$$

where $f:\{0,1,\ldots,n\}\times S \to \mathbb R^n$ is an injective mapping, and $\widehat t_i(k,s)$ denotes the $i$th component of $f(k,s)$. Since the adjustment transfers may scale like $1/\delta$ as $\delta\to0$, the image of $f$ can be chosen to be disjoint from the set of possible adjustment transfers. Thus, conditional on the high-probability event, observing $(k,\widehat{\mathbf t})$ is equivalent to observing $(k,s)$. Hence both observables induce exactly the same posterior beliefs. With the remaining probability $\delta$, it assigns to each agent $i$ an adjustment transfer $\widetilde T_i(v^\prime)$, chosen so that each agent's interim expected transfer exactly matches that induced by $\Gamma^s$. Namely, define\footnote{As $\delta\to 0$, the realized value of $\widetilde T_i(v^\prime)$ on the event of probability $\delta$ may scale like $1/\delta$. This reflects a standard rare-event / large-transfer construction: expected transfers are matched while the probability of informative transfer realizations becomes arbitrarily small. Related low-probability, high-magnitude transfer devices appear in other mechanism design settings (e.g., \citealp{CremerMcLean1988}). If transfers are required to remain bounded (e.g., under limited liability), fixing $\delta$ away from zero or truncating $\widetilde T_i(v^\prime)$ yields bounded transfers and an arbitrarily close approximation, though not exact convergence as $\delta\to 0$.}

\begin{equation*}
\widetilde T_i(v^\prime)  =  \frac{T_i(v^\prime) -  (1-\delta)  \mathbb E_{V_{-i}}  \left[\widehat t_i(k(v^\prime,V_{-i}),s(v^\prime,V_{-i}))  \right]}{\delta}.
\end{equation*}
Importantly, $\widetilde T_i(v^\prime)$ depends only on $v^\prime$ and not on the agent's true type $v$. By construction, the interim expected transfer of agent $i$ following report $v^\prime$ is exactly

\begin{equation*}
(1-\delta)\mathbb E_{V_{-i}}\left[\widehat t_i(k(v^\prime,V_{-i}),s(v^\prime,V_{-i}))\right]  +  \delta \widetilde T_i(v^\prime)  =  T_i(v^\prime).
\end{equation*}

Hence the public transfers mechanism reproduces the interim expected transfers induced by $\Gamma^s$ exactly.

On the event of probability $1-\delta$, observing $(k,\widehat{\mathbf t})$ induces exactly the same posterior beliefs as observing $(k,s)$ in $\Gamma^s$, so continuation payoffs coincide with those in the signaling mechanism. On the complementary event, which has probability $\delta$, observing the public transfer vector may induce different posterior beliefs and therefore different continuation payoffs. Let $\widehat U_i^\delta(v,v^\prime)$ denote agent $i$'s interim payoff in the public transfers mechanism, and let $R_i^\delta(v,v^\prime)$ denote the expected continuation payoff conditional on the event of probability $\delta$. Then,
\begin{equation*}
\widehat U_i^\delta(v,v^\prime) =  v Q_i(v^\prime) - T_i(v^\prime)  + (1-\delta) \mathbb E_{V_{-i}}\left[\Phi_i\left(v,k(v^\prime,V_{-i}), \mathbf F(k(v^\prime,V_{-i}),s(v^\prime,V_{-i}))\right)\right] + \delta R_i^\delta(v,v^\prime).
\end{equation*}

Since $V_i$ is compact, $K$ is finite, and $\Delta(V_1\times\cdots\times V_n)$ is compact under the weak topology, continuity of $\Phi_i$ implies uniform boundedness. Hence there exists a constant $C<\infty$ such that, for all $i,v,k,\mathbf F$, $|\Phi_i(v,k,\mathbf F)|\leq C$. In particular, $|R_i^\delta(v,v^\prime)|\le C$ for all $i,v,v^\prime$. Therefore, for all $i,v,v^\prime$,
\begin{equation*}
\left|\widehat U_i^\delta(v,v^\prime) - U_i(v,v^\prime)\right|\leq 2C\delta.
\end{equation*}

It follows that truthful reporting is a $4C\delta$-Bayes-Nash equilibrium of the public transfers mechanism. Indeed, for every deviation $v^\prime$,
\begin{eqnarray*}
\widehat U_i^\delta(v,v) &\geq & U_i(v,v) - 2C\delta  \\
&\geq & U_i(v,v^\prime) - 2C\delta  \\
&\geq & \widehat U_i^\delta(v,v^\prime) - 4C\delta.
\end{eqnarray*}

Choose $\delta\in(0,1)$ such that $4C\delta<\varepsilon$. Then truthful reporting is an $\varepsilon$-Bayes-Nash equilibrium of the constructed public transfers mechanism. Moreover, interim expected payoffs differ from those in $\Gamma^s$ by at most $2C\delta<\varepsilon$, while allocation probabilities and expected transfers coincide exactly with those induced by $\Gamma^s$. Hence the constructed public transfers mechanism satisfies the properties stated in part (i) of the theorem. This proves approximate implementability with public transfers.

\end{proof}

The main idea behind the proof is as follows. In the signaling mechanism, each type $v$ faces an interim expected transfer $T_i(v)$ and continuation payoffs determined by the information revealed through $(k,s)$. Under public observability of transfers, directly reproducing the same transfer rule would generally reveal too much information. The construction therefore separates the informational and incentive roles of transfers.

The constructed public transfers mechanism is direct. Given a reported type profile, it reproduces the allocation and signal induced by the equilibrium of the signaling mechanism. With probability $1-\delta$, the mechanism assigns transfers of the form $\widehat t_i(k,s)=f_i(k,s)$, where the injective mapping $f$ encodes exactly the same informational content as the pair $(k,s)$ observed
in the signaling mechanism. Hence, on this event, posterior beliefs and continuation payoffs coincide exactly with those induced by the signaling mechanism. With the remaining probability $\delta$, the mechanism assigns an adjustment transfer $\widetilde T_i(v)$ chosen so that interim expected transfers exactly match those induced by the equilibrium of $\Gamma^s$. Large transfers arise precisely because the informationally equivalent transfer encoding $\widehat t_i(k,s)$ generally does not reproduce the interim expected transfers induced by the signaling mechanism. The adjustment transfer $\widetilde T_i(v)$ compensates for this discrepancy while occurring only on an event of vanishing probability.

Additional information may only be revealed on an event whose probability is arbitrarily small. Since continuation payoffs are uniformly bounded, the resulting distortion in interim expected payoffs and incentives is at most of order $\delta$. The proof therefore does not require tracking the exact posterior beliefs induced on the low-probability event: uniform boundedness of continuation payoffs is sufficient to control the resulting utility distortion. By choosing $\delta$ sufficiently small, the constructed public transfers mechanism approximates arbitrarily well the
equilibrium outcome induced by the signaling mechanism.

\medskip

Theorem \ref{prop1} connects the present framework to several existing approaches to mechanism design with signaling, aftermarkets and belief-dependent preferences.

\begin{example}[Signaling in auctions, \cite{BosPollrich}]
Consider $\Phi_i(v_i,k,\mathbf F)=\lambda\,\mathbb E[V_i\mid \mathbf F]$ with $\lambda>0$, together with the class of auction mechanisms in which $k=i$ whenever $t_i=\max_{j=1,\dots,n} t_j$, and any set of signals $S$. In this environment, the approximate equivalence of Theorem \ref{prop1} strengthens to an exact equivalence: any transfer rule $t$ that is optimally implementable in a signaling auction is also optimally implementable in a public transfers auction.
\end{example}

\begin{rem}[Belief-dependent preferences, \cite{ErnestoMora}]
\cite{ErnestoMora} establish a revelation principle for settings where agents have belief-dependent preferences. Our result complements this approach by showing that approximate implementability can be preserved even when transfers become publicly observable, while preserving the convenient analytical shortcut of identifying each type with its expected transfer.
\end{rem}

In \cite{Dworczak19}, implementability is characterized by the information partition induced by the mechanism, consisting of the cutoff structure and the post-allocation signals. Theorem \ref{prop1} shows that this informational structure can be approximated arbitrarily closely using publicly observed transfers instead of explicit signals. The construction relies on the $\varepsilon$-mixing device developed in the proof of Theorem \ref{prop1}, which ensures that any additional information conveyed by observable transfers arises only on a vanishing-probability event.

\begin{example}[Illustrative example: approximate replication with public transfers]

We consider two agents $i\in\{1,2\}$ with valuations $v_i\in\{L,H\}$ drawn i.i.d. and $0<L<H$ for a single object to allocate.\footnote{As in the preliminaries to establish Theorem \ref{prop1}, we do not need to provide more details about the distribution. Yet, characterizing the equilibrium and the incentive constraints would require specifying a distribution over types, for example $\Pr(V_i=H)=p\in(0,1)$.} Agent $i$'s utility from the allocation $(k,\mathbf t)$, with $K=\{0,1,2\}$ is given by
$$
v_i\mathds 1_{\{k=i\}}  - t_i + \alpha\,\mathbb{E}[V_i\mid \mathcal O], \text{ with }\alpha\ge 0,
$$

and where $\Phi_i(v_i,k,\mathbf F)= \alpha\,\mathbb{E}[V_i \mid \mathcal O]$. Therefore, rather than explicitly representing posterior beliefs through the distribution $\mathbf F$, we directly refer to the information observable $\mathcal O=(k,s)$ available to the aftermarket.

The social choice function considered here is such that the mechanism implements an efficient allocation (the higher valuation agent wins, ties are broken fairly) with the winner's transfer given by

$$
t(H)=\tau_H,\;\; t(L)=\tau_L,\;\; 0\le \tau_L<\tau_H<H,
$$

while loser transfers remain equal to $0$. Let us choose $(\tau_L,\tau_H)$ to satisfy the incentive compatibility (IC) and individual rationality (IR) constraints in this environment.

\paragraph{Signaling Mechanism $\Gamma^s$.}
Let us consider mechanism $\Gamma^s$ as described in Section \ref{Section:Model}. Each agent sends a message $m_i\in\{L,H\}$. The object is allocated to the highest report, while transfers are determined by $t(m_i)\in\{\tau_L,\tau_H\}$ for the winner and $0$ for the loser. The signal sent to the aftermarket is $s\in\{0,1\}$ with $s=1$ if and only if $t_i(m_i)=\tau_H$. The public observes $(k,s)$ but neither transfers nor messages.

Let $Q(v^\prime)$ be the interim winning probability when reporting $v^\prime\in\{L,H\}$ and $T(v^\prime)$ the expected transfer. For a type $v$ reporting $v^\prime$, expected utility is
$$
U_i(v,v^\prime)=v\,Q(v^\prime)-T(v^\prime)+\alpha\,\mathbb E_{V_{-i}}\left[\mathbb E[V_i\mid\mathcal O]\right],
$$
with $\mathcal O=(k(v^\prime,V_{-i}), s(v^\prime,V_{-i}))$. As the signal $s$ is binary, the winner's type is revealed as follows: $(k=i,s=1)\Rightarrow v_i=H$ and $(k=i,s=0)\Rightarrow v_i=L$.

\paragraph{Public Transfers Mechanism.}
We implement the same allocation but without an explicit signal. Instead, we make the observed transfer play the role of the signals by making the entire allocation $(k,\mathbf t)$ public. 

\medskip
\textbf{Signal-to-transfer encoding.} First, we directly map the signal $s$ from the signaling mechanism $\Gamma^s$ into the observed transfer deterministically, without adding any extra information.

\medskip

Let us choose an injective mapping $f:K\times\{0,1\}\to\mathbb R$, e.g.
$$
f(k,s)=10\cdot\mathds 1_{\{k=1\}}+20\cdot\mathds 1_{\{k=2\}}+s.
$$
Define the publicly observed transfer for the winner $i$ as
$$
\tilde t_i(k) \;=\; f\!\left(k,\,\mathds 1_{\{t_i=\tau_H\}}\right),
$$
with losers paying $0$. Because $f$ is injective, observing $(k,\tilde t_i)$ reveals exactly the same signal $s$ as in signaling mechanism $\Gamma^s$. Hence the aftermarket forms exactly the same posterior beliefs from $(k,\tilde t_i)$ as from $(k,s)$ in the signaling mechanism. This step illustrates how publicly observed transfers can replicate the informational role of signals.

\medskip

\textbf{$\varepsilon$-mixing device.}
We now introduce a small amount of randomization to prevent observable transfers from revealing additional information beyond the informational content of $(k,s)$.

\medskip

Fix $\delta\in(0,1)$. With probability $1-\delta$, the mechanism assigns the same publicly observed transfers as in the signal-to-transfer encoding step, namely
$$
\widehat t_i(k,s)=\widetilde t_i(k),
$$
with losers continuing to pay $0$. Hence observing $(k,\widehat t_i)$ induces exactly the same posterior beliefs as observing $(k,s)$ in the signaling mechanism. With the remaining probability $\delta$, the mechanism assigns an adjustment transfer $\widetilde T_i(v)$ chosen so that interim expected transfers coincide with those in the signaling mechanism.

Any additional information is therefore revealed only on an event of probability $\delta$. Since continuation payoffs are uniformly bounded, the resulting distortion in interim expected payoffs and incentives is of order $\delta$. Hence, by choosing $\delta$ sufficiently small, the public transfers mechanism approximates arbitrarily well the equilibrium outcome induced by the signaling mechanism.
\end{example}

\section{Conclusion} \label{Section:Conclusion}

Our analysis highlights a novel perspective on implementability and transfer observability in mechanism design. While related work studies how transfer observability shapes optimal mechanisms through information disclosure, our focus is on a more structural question: whether transfer observability fundamentally restricts implementability itself. We show that, under continuation payoffs that depend on observables only through posterior beliefs, the distortions induced by observable transfers can be made arbitrarily small. Specifically, any equilibrium outcome implementable with hidden transfers can be approximated arbitrarily well with publicly observable transfers. Observable transfers therefore generate at most a vanishing implementability distortion. Understanding when such approximations can be achieved with uniformly bounded transfers remains an interesting open question. These findings provide a new perspective on signaling in mechanisms followed by an aftermarket or more generally by post-mechanism belief formation. The framework also accommodates continuation payoffs arising from strategic responses by outsiders, provided these responses depend on observables only through the posterior beliefs they induce.

Future work may further explore how this framework applies across domains such as procurement, takeovers, advertising, fundraising, and other real-world contexts in which transfer transparency plays an important role.\footnote{See, for example, \cite{benabou2006incentives} on social image in charitable mechanisms, \cite{holmstrom1999managerial} on takeovers, and \cite{rayo2010optimal} on consumer choice.} An important direction for future research is to study environments in which continuation payoffs depend on observables through channels other than posterior beliefs.
\bibliography{biblio}
\end{document}